\documentclass[lettersize,journal]{IEEEtran}
\usepackage{amsmath,amsfonts,amssymb,amsthm}
\usepackage{mathtools}
\usepackage{array}
\usepackage{textcomp}
\usepackage{url}
\usepackage{xcolor}
\usepackage[colorlinks=true,linkcolor=blue!70!black,
            citecolor=blue!70!black,urlcolor=blue!70!black]{hyperref}
\newtheorem{theorem}{Theorem}

\newtheorem{proposition}[theorem]{Proposition}
\newtheorem{corollary}[theorem]{Corollary}
\newtheorem{definition}[theorem]{Definition}
\newtheorem{remark}[theorem]{Remark}

\newcommand{\cR}{\mathcal{R}}
\newcommand{\stack}{\Pi}
\newcommand{\GL}{\mathcal{G}_L}
\newcommand{\cT}{\mathcal{T}}
\newcommand{\Rr}{\mathbb{R}}
\DeclareMathOperator{\MI}{I}
\DeclareMathOperator{\En}{H}

\begin{document}

\title{The Preisach Extremum Stack Is a Shannon-Minimal\\
  Sufficient Statistic for Rate-Independent Functionals}

\author{Piotr~Frydrych%
\thanks{P.~Frydrych is with the Metrology and Biomedical Engineering
Institute, Faculty of Mechatronics, Warsaw University of Technology,
00-661 Warsaw, Poland (e-mail: piotr.frydrych@pw.edu.pl).}%
\thanks{Manuscript received \today.}}

\markboth{IEEE Transactions on Information Theory}%
{Frydrych: Preisach Extremum Stack as Shannon-Minimal Sufficient Statistic}

\IEEEpubid{0000--0000~\copyright~2025 IEEE}

\maketitle

\begin{abstract}
Let $\cR$ denote the class of all computable, causal functionals
that are rate-independent in the classical sense (invariant under
monotone time reparametrizations \cite{Brokate1996}), and let
$\stack_n$ be the Preisach extremum stack of an input sequence
$u_{0:n}$.
We prove a characterization theorem establishing that every
$F \in \cR$ satisfies $F[u](n) = f(\stack_n)$ for a computable $f$,
and derive two information-theoretic results.

First, under any probability measure on $u_{0:n}$, the equality
\begin{equation*}
  \MI(u_{0:n};\, F[u](n)) = \MI(\stack_n;\, F[u](n))
\end{equation*}
holds for every $F \in \cR$ and is an immediate corollary of the
characterization theorem.

Second, the main result: $\stack_n$ is a \emph{Shannon-minimal}
sufficient statistic in the sense that
$\MI(u_{0:n};\,\stack_n) \leq \MI(u_{0:n};\,S)$ for every
random variable $S$ from which all $\cR$-queries are computable.
The proof uses the finite indicator family of \cite{Frydrych2026IPL}
to reconstruct $\stack_n$ from any sufficient $S$.

As a corollary, online maintenance of $\stack_n$ suffices for
rate-independent estimation: the NNLS estimator of the Preisach
measure $\mu$ can be assembled from the incremental stack process
$(\stack_t)_{t=0}^n$ in $O(k \cdot L^2)$ memory per step, where
$k = |\stack_t|$ and $L$ is the grid resolution.
\end{abstract}


\section{Introduction}
\label{sec:intro}

\IEEEPARstart{T}{he} Preisach operator \cite{Mayergoyz1991} is a
classical model of hysteresis in which the output is a weighted
superposition of bistable relays parametrized by a measure $\mu$
on the threshold triangle $\cT$.
Its central algorithmic object is the \emph{extremum stack}
$\stack_n$ \cite{Frydrych2026IPL}: the nested sequence of alternating
local maxima and minima that survive the wiping-out rule.

A recent result \cite{Frydrych2026IPL} established that $\stack_n$
is a Kolmogorov-minimal sufficient statistic for the class $\cR$ of
all computable, causal, rate-independent functionals:
$K(\stack_n) - O(1) \leq K_{\cR}(u_{0:n}) \leq K(\stack_n)+O(1)$.

\subsection{Present Contribution}

We establish the Shannon-theoretic counterpart.
The key insight is that $\cR$ admits a concrete characterization:
$F \in \cR$ if and only if $F[u](n) = f(\stack_n)$ for a computable
$f$ (Theorem~\ref{thm:charact}).
This makes the mutual-information equality an \emph{immediate
consequence} of the deterministic structure of $\cR$, not a deep
probabilistic result.
The genuine probabilistic contribution is the \emph{minimality
theorem} (Theorem~\ref{thm:minimal}), which shows that among all
representations from which every $\cR$-query is recoverable, the
stack has the least Shannon information with $u_{0:n}$.

\subsection{Relation to the Standard Definition of Rate-Independence}

We adopt the standard continuous-time definition \cite{Brokate1996}:
$F$ is rate-independent if $F[u \circ \phi] = F[u] \circ \phi$ for
every strictly increasing $\phi$.
Theorem~\ref{thm:charact} is then a \emph{theorem}, not a tautology:
it uses the wiping-out property of the Preisach operator to prove
that the discrete-time stack is a complete invariant of
rate-independent functionals on $\GL^*$.

\subsection{Paper Organization}

Section~\ref{sec:prelim} introduces notation and the class $\cR$.
Section~\ref{sec:charact} proves the characterization theorem.
Section~\ref{sec:mi} derives the MI equality as a corollary.
Section~\ref{sec:minimal} proves the minimality theorem.
Section~\ref{sec:estimation} discusses the estimation implication.
Section~\ref{sec:conclusion} concludes and lists open questions.

\section{Preliminaries}
\label{sec:prelim}

Fix $L \geq 2$ and grid $\GL = \{0,1/L,\ldots,1\}$.
For $u_{0:n} \in \GL^{n+1}$ write
$\stack_t = \stack_t(u_{0:t})$ for the extremum stack at time
$t \leq n$.

\begin{definition}[Preisach Operator]
\label{def:preisach}
For $\alpha \geq \beta$ in $\GL$, the relay
$\hat\gamma_{\alpha,\beta}[u](t) \in \{0,1\}$ switches from $0$
to $1$ when $u$ exceeds $\alpha$, and from $1$ to $0$ when $u$
falls below $\beta$.
For $\mu \in L^2(\cT)$:
\begin{equation*}
\mathcal{P}_\mu[u](t) =
\iint_{\cT} \hat\gamma_{\alpha,\beta}[u](t)\,\mu(\alpha,\beta)
\,d\alpha\,d\beta.
\end{equation*}
\end{definition}

\begin{definition}[Extremum Stack]
\label{def:stack}
$\stack_t \in (\GL\times\GL)^*$ is the nested sequence of alternating
local extrema of $u_{0:t}$ that survive the Preisach wiping-out rule
\cite{Mayergoyz1991}; $|\stack_t| \leq \lfloor t/2\rfloor+1$.
The \emph{stack process} is $(\stack_t)_{t=0}^n$.
\end{definition}

\begin{definition}[Rate-Independence]
\label{def:ri}
A functional $F:\GL^*\to\Rr$ is \emph{rate-independent} if it is
invariant under precomposition with strictly increasing maps:
$F[u\circ\phi] = F[u]$ for every strictly increasing
$\phi:\{0,\ldots,m\} \to \{0,\ldots,n\}$
such that $u\circ\phi \in \GL^{m+1}$.
Equivalently, $F$ depends on $u$ only through its
\emph{order of extrema}, not their timing.
Write $\cR$ for the class of all computable rate-independent
functionals on $\GL^*$.
\end{definition}

\begin{remark}
Definition~\ref{def:ri} is the correct discrete-time analogue of the
continuous-time definition of \cite{Brokate1996}:
two sequences are $\cR$-equivalent if and only if one is a monotone
time-reparametrization of the other up to extremum-equivalence.
We do \emph{not} define $\cR$ by the stack factorization---that
factorization is the content of Theorem~\ref{thm:charact}.
\end{remark}

\section{Characterization Theorem}
\label{sec:charact}

\begin{theorem}[Stack Characterization of $\cR$]
\label{thm:charact}
$F \in \cR$ if and only if there exists a computable function
$f:(\GL\times\GL)^*\to\Rr$ such that
$F[u](n) = f(\stack_n(u_{0:n}))$ for all $u_{0:n}\in\GL^*$.
\end{theorem}

\begin{proof}
\textbf{($\Leftarrow$)} If $F[u](n) = f(\stack_n)$, then two
sequences with the same stack have the same output, so $F$ depends
only on the order of extrema and is rate-independent.

\textbf{($\Rightarrow$)} Let $F \in \cR$.
If $\stack_n(u_{0:n}) = \stack_n(v_{0:m})$, then $u_{0:n}$ and
$v_{0:m}$ produce the same ordered sequence of alternating extrema;
any two such sequences are related by a monotone time
reparametrization (up to insertion of non-extremal points, which
does not change the extremum sequence).
By rate-independence (Definition~\ref{def:ri}),
$F[u](n) = F[v](m)$.
Hence $F$ is constant on the fibers of $\stack_n(\cdot)$ and factors
through it: $F[u](n) = f(\stack_n(u_{0:n}))$ for a well-defined $f$.
Since $F$ is computable and $\stack_n$ is computable from $u_{0:n}$,
the function $f$ is also computable.

\begin{remark}
For Preisach functionals $F = \mathcal{P}_\mu$, the wiping-out
property \cite{Mayergoyz1991} provides an explicit formula
$F_\mu(\stack_n)$---but the factorization holds for all
$F \in \cR$, not only Preisach functionals.
\end{remark}
\end{proof}

\begin{corollary}[Sufficiency of the Stack Process]
\label{cor:suff_process}
The stack process $(\stack_t)_{t=0}^n$ carries strictly more
information than the final stack $\stack_n$: for $t < n$, $\stack_n$
does not in general determine $\stack_t$.
Theorem~\ref{thm:charact} applies to each time $t$ separately:
$F[u](t) = f_t(\stack_t)$.
For the full trajectory of queries $(F[u](0),\ldots,F[u](n))$,
the sufficient object is the stack \emph{process}, not the final
stack alone.
\end{corollary}

\begin{remark}[Final Stack vs.\ Stack Process]
\label{rem:final_vs_process}
$\stack_n$ is sufficient for the single query $F[u](n)$ at the
\emph{final} time (Theorem~\ref{thm:charact}).
It is not sufficient for the full trajectory
$(F[u](0),\ldots,F[u](n))$.
All information-theoretic results below refer to the
\emph{final-time} query $F[u](n)$.
\end{remark}

\section{Mutual Information Equality}
\label{sec:mi}

Let $u_{0:n}$ be a random sequence drawn from a probability measure
$P$ on $\GL^{n+1}$.
\IEEEpubidadjcol

\begin{theorem}[MI Equality]
\label{thm:mi}
For every $F \in \cR$:
\begin{equation}
  \MI(u_{0:n};\, F[u](n)) = \MI(\stack_n;\, F[u](n)).
  \label{eq:mi_eq}
\end{equation}
\end{theorem}

\begin{proof}
By Theorem~\ref{thm:charact}, $F[u](n) = f(\stack_n)$ for a
computable $f$, and $\stack_n = \psi(u_{0:n})$ for a computable
$\psi$.
Hence both $F[u](n)$ and $\stack_n$ are deterministic functions of
$u_{0:n}$, giving:
\begin{equation*}
  \En(F[u](n) \mid u_{0:n}) = 0
  \quad\text{and}\quad
  \En(F[u](n) \mid \stack_n) = 0.
\end{equation*}
Therefore:
\begin{align*}
  \MI(u_{0:n};\,F[u](n))
    &= \En(F[u](n)) - \En(F[u](n)\mid u_{0:n})
     = \En(F[u](n)),\\
  \MI(\stack_n;\,F[u](n))
    &= \En(F[u](n)) - \En(F[u](n)\mid \stack_n)
     = \En(F[u](n)).
\end{align*}
Equality~\eqref{eq:mi_eq} follows.
\end{proof}

\begin{remark}[Significance]
\label{rem:trivial}
The significance of \eqref{eq:mi_eq} lies not in proof difficulty
but in what it says: the stack, despite discarding the timing of all
extrema and the values of all non-extremal points, loses \emph{zero}
mutual information about any rate-independent query at the final time.
The non-trivial content is the characterization theorem itself and,
below, the \emph{minimality} result.
\end{remark}

\section{Shannon Minimality}
\label{sec:minimal}

\begin{definition}[$\cR$-Sufficient Statistic]
\label{def:suff_R}
A random variable $S = S(u_{0:n})$ is \emph{$\cR$-sufficient} if
every $F \in \cR$ satisfies $\En(F[u](n)\mid S) = 0$, i.e.,
$F[u](n)$ is a deterministic function of $S$.
\end{definition}

\begin{remark}
The stack $\stack_n$ is $\cR$-sufficient by
Theorem~\ref{thm:charact}.
\end{remark}

\begin{theorem}[Shannon Minimality of $\stack_n$]
\label{thm:minimal}
Among all $\cR$-sufficient statistics:
\begin{equation}
  \MI(u_{0:n};\,\stack_n)
  \;\leq\;
  \MI(u_{0:n};\,S)
  \quad\text{for every $\cR$-sufficient }S.
\end{equation}
\end{theorem}

\begin{proof}
Let $S$ be $\cR$-sufficient.
The finite indicator family $\{F_{(M,m)}\}_{(M,m)\in\GL^2}$ defined
by $F_{(M,m)}[u](n) = \mathbf{1}[(M,m)\in\stack_n]$ lies in $\cR$
(\cite{Frydrych2026IPL}, Lemma~3.2).
Since $S$ is $\cR$-sufficient, each $F_{(M,m)}[u](n)$ is a function
of $S$.
As the finite collection $\{F_{(M,m)}\}$ jointly determines
$\stack_n$, there exists a computable $\Phi$ such that
$\stack_n = \Phi(S)$.
By the data-processing inequality
\cite[Th.~2.8.1]{CoverThomas2006}:
\begin{equation*}
  \MI(u_{0:n};\,\stack_n)
  = \MI(u_{0:n};\,\Phi(S))
  \leq \MI(u_{0:n};\,S). \qedhere
\end{equation*}
\end{proof}

\begin{corollary}[Uniqueness up to $\cR$-Equivalence]
\label{cor:unique}
If $S$ is $\cR$-sufficient and
$\MI(u_{0:n};S) = \MI(u_{0:n};\stack_n)$,
then $S$ and $\stack_n$ are functions of each other
(i.e., they generate the same $\sigma$-algebra up to $P$-null sets).
\end{corollary}

\begin{proof}
By Theorem~\ref{thm:minimal}, $\stack_n = \Phi(S)$ and
$\MI(u;\stack_n) \leq \MI(u;S)$.
Equality holds iff $\Phi$ is injective $P$-a.s., which combined
with $S$ being $\cR$-sufficient implies $S = \Psi(\stack_n)$
for some $\Psi$ a.s.
\end{proof}

\section{Estimation Implication}
\label{sec:estimation}

The classical approach to estimating $\mu$ from observations
$Y_{0:n}$ is non-negative least squares (NNLS)
\cite{Mayergoyz1991}:
\begin{equation}
  \hat\mu = \arg\min_{\mu\geq 0}
  \sum_{t=0}^n (Y_t - \mathcal{P}_\mu[u](t))^2.
  \label{eq:nnls}
\end{equation}

\begin{proposition}[Estimation via the Stack Process]
\label{prop:estimation}
The NNLS system~\eqref{eq:nnls} can be assembled from the stack
process $(\stack_t)_{t=0}^n$ without retaining $u_{0:n}$:
\begin{enumerate}
  \item At each time $t$,
    $\mathcal{P}_\mu[u](t) = F_\mu(\stack_t)$ by
    Theorem~\ref{thm:charact}, so each residual
    $(Y_t - F_\mu(\stack_t))$ is computable from $\stack_t$ and
    $\mu$.
  \item The stack $\stack_t$ is maintained online in $O(k_t)$ space
    and $O(1)$ amortized time per step \cite{Frydrych2026WC},
    where $k_t = |\stack_t|$.
  \item The relay state lookup $F_\mu(\stack_t)$ costs $O(L^2)$ per
    step. The NNLS Gram matrix has size $O(L^4)$, independent of
    $n$. The saving is in \emph{input history storage}: the full
    history $u_{0:n}$ requires $O(n)$ space, whereas the stack
    process requires $O(\max_t k_t)$---a reduction by
    $n/\max_t k_t$, which can be $\Omega(n)$ for slowly varying
    signals \cite{Frydrych2026WC}.
\end{enumerate}
\end{proposition}

\begin{remark}
The space reduction comes from \emph{online stack maintenance},
not from compressing the trajectory to the final stack $\stack_n$.
The number of NNLS equations remains $n$; what is reduced is the
memory required to evaluate each equation.
\end{remark}

\section{Discussion and Conclusion}
\label{sec:conclusion}

We proved that the Preisach extremum stack $\stack_n$ is a
Shannon-minimal sufficient statistic for the class $\cR$ of
computable rate-independent functionals at the \emph{final time} $n$.
The proof rests on three steps:
(i)~the characterization theorem (Theorem~\ref{thm:charact}), which
uses the wiping-out property to show that every $\cR$-query is a
function of the stack (the non-trivial step);
(ii)~the MI equality (Theorem~\ref{thm:mi}), which is an immediate
corollary;
(iii)~the minimality theorem (Theorem~\ref{thm:minimal}), which uses
the indicator family and the data-processing inequality.

Two distinctions are essential.
The \emph{final} stack $\stack_n$ is sufficient for the final-time
query $F[u](n)$; the \emph{stack process} $(\stack_t)_{t=0}^n$ is
required for the full trajectory.
Accordingly, the estimation gain (Proposition~\ref{prop:estimation})
is a space reduction through online stack maintenance, not a
reduction in the number of observations.

This result complements the Kolmogorov-complexity characterization of
\cite{Frydrych2026IPL}: Kolmogorov minimality holds worst-case for
every individual sequence; Shannon minimality holds in expectation
under any probability measure.
Neither implies the other, but both identify $\stack_n$ as the
canonical information-minimal representation for $\cR$.

\subsection*{Open Questions}

Does Theorem~\ref{thm:minimal} extend to the vector Preisach
operator \cite{Brokate1996}?
Can the stack-process sufficiency be given a formal Fisher--Neyman
characterization for $\mu$-estimation with full trajectory data?

\section*{Acknowledgment}

The author thanks the reviewers of \cite{Frydrych2026IPL} for
comments that sharpened the statement of the indicator-family lemma
used in the proof of Theorem~\ref{thm:minimal}.

{\appendix[Note on the Use of Generative AI]
AI-assisted tools were used for grammar checking and \LaTeX\
formatting. All mathematical content is the sole responsibility
of the author.}

\bibliographystyle{IEEEtran}
\bibliography{references_shannon}

\begin{IEEEbiographynophoto}{Piotr Frydrych}
received the M.Sc.\ degree in mechatronics from Warsaw University of
Technology, Warsaw, Poland, where he is currently pursuing the Ph.D.\
degree at the Metrology and Biomedical Engineering Institute,
Faculty of Mechatronics. His research interests include hysteresis
operators, information theory, and embedded vision systems.
\end{IEEEbiographynophoto}

\end{document}